\documentclass[14pt,a4paper]{article}
\usepackage{fancyhdr}
    \usepackage[
    top    = 2.75cm,
    bottom = 2.50cm,
    left   = 3.00cm,
    right  = 2.50cm]{geometry}
\usepackage{pstricks-add}   
\usepackage{graphics,graphicx}
\usepackage{pst-node,pst-tree}
\usepackage{enumerate,amsmath,amsthm,mathrsfs,array,pstricks,amssymb}
\usepackage{auto-pst-pdf}
\usepackage{algpseudocode}
\usepackage{algorithm}
\usepackage{hyperref}
\theoremstyle{bthm}

\newtheorem{thm}{Theorem}[section]
\newtheorem{lemma}{Lemma}[section]
\newtheorem{definition}[thm]{Definition}

\date{}
\begin{document}
\title{Equiseparability on Terminal Wiener Index and Distances}
\author{Sulphikar A}
\maketitle
\begin{center}
\textbf{ Department of Computer Science and Engineering\\ National Institute of Technology , Trichy\\ Tamilnadu INDIA}
\end{center}
\begin{abstract}
Teminal Wiener index is one of the commonly used topological index in mathematical chemistry. If two or more chemical compounds have the same  terminal Wiener index then they will have similar physico-chemical properties. In this work we propose a new method  for constructing equiseparable  trees w.r.t  terminal Wiener index. The existing method is based on the number of pendent vertices but the proposed  method is based on distance parameters. 
\end{abstract}
\textbf{keywords}\\
Terminal Wiener index, Wiener index, Distance in graphs, Equiseparability.
\section{Introduction}
Given a graph $G=(V,E)$, the Wiener index  $W(G)$ of $G$ is defined  as
\begin{equation}\label{eq:1}
W(G)=\sum\limits_{(u,v)\in V} d(u,v).
\end{equation}
 where $d(u,v)$ is the distance between vertices $u$ and $v$ of $G$.\cite{chen,dobr}\\
Let $T$ be an $n$-vertex tree and $e=uv$ an edge of $T$. Denote by $n_{u}(e|T)$(resp. $n_{v}(e|T)$) the number of vertices of $T$ lying on one side of the edge $e$, closer to vertex $u$(resp. $v$). Then $W(T)$ also satisfies the following relation:
\begin{equation}\label{eq:2}
 W(T)=\sum\limits_{e=uv} n_{u}(e|T)n_{v}(e|T)
\end{equation}
 Let $T$ be an $n$-vertex tree with $k$ pendent vertices. Then the terminal Wiener index $TW(T)$ \cite{zery,gut,surv} of $T$ is defined as the sum of the distances between  all pairs of pendent vertices of $T$. That is
\begin{equation}\label{eq:3}
 TW(T)=\sum\limits_{1 \leq i <j \leq k} d(v_{i},v_{j})
\end{equation}
   where $v_{i}$,$v_{j}$ are pendent vertices in  $T$. Let $P(T)$ be the set of pendent vertices of $T$. 
	Terminal Wiener index can also be calculated as
\begin{equation}\label{eq:4}
TW(T)=\frac{1}{2}\sum\limits_{u\in P(T)} d^{+}(u)
\end{equation}
where $d^{+}(u)$  is the sum of the distances between $u$ and pendent vertices of $T$.
For an edge $e=uv$, $p_{u}(e)$(resp. $p_{v}(e)$) denotes the number of vertices closer to $u$(resp. $v$) than $v$(resp. $u$). The numbers  $p_{u}(e)$ and $p_{v}(e)$  satisfy the relation  
\begin{equation}\label{eq:5}
d^{+}(u)-d^{+}(v)=p_{v}(e)-p_{u}(e)
\end{equation}
Let $P_{n}$ be a path on $n$ vertices. It is easy to see that $d^{+}(v)=n-1$ for any $v \in V(P_{n})$. Let $uv$ be a pendent edge of a tree $T$ with $deg(u)=1$. Then $d^{+}(u)-d^{+}(v)=p-2$. For a star $S_{n}$ on $n$ vertices, $d^{+}(v)=n-1$ if $v$ is the nonpendent vertex and $d^{+}(v)=2(n-2)$ otherwise.
\begin{definition}\cite{xia}
Assuming that $p_{u}(e) \leq p_{v}(e)$, two trees T' and T'' of order $n$  with the same number of pendent vertices are said to be equiseparable w.r.t  terminal Wiener index if their edges $e'_{1},e'_{2},....e'_{n-1}$ and $e''_{1},e''_{2},....e''_{n-1}$ can be labelled such that the equality $p_{u}(e'_{i}|T')=p_{u}(e''_{i}|T'')$ holds for all $i=1,2,....n-1$. $p_{u}(e'_{i}|T')$ denotes $p_{u}(e'_{i})$  in $T'$.
\end{definition}
The concept of Wiener index has been extensively used in chemistry. If two or more chemical compounds  have equiseparable trees, then those compounds will have similar physico-chemical properties which can not be distinguished by means of indices like Wiener index and terminal Wiener index\cite{xia}.
Terminal distance matrices were used in the mathematical modelling of proteins and genetic codes\cite{hor}.\\
One way of generating equiseparable  graphs  with respect to  Wiener  index can be found  in \cite{aado}. Construction of a class of   equiseparable trees with respect to terminal Wiener  index can be found in \cite{xia}. A detailed study of equiseparable molecules with respect to Wiener index can be found in \cite{gut1,surv,xia2}. \\ 
It is well known that   almost all trees have a terminal equiseparable mate.\cite{xia}\\
\noindent\underline{Our contribution:}\\
In this work, we propose  a new method to generate trees having the same terminal Wiener index. This method is based on distances between vertices. Using this method it is possible to construct many tree pairs with same terminal Wiener index.\\

This paper is organized as follows:  In Section $2$ we  explain some methods  for construction of equiseparable  tree pairs with respect to  terminal Wiener index only. In Section $3$ we  explain some methods  for constructing large class  of trees that are equiseparable  with respect to terminal Wiener index. 
\section{Construction of equiseparable  trees with respect to terminal Wiener Index}
 We use the notation $d^{+}(u)$ to denote the sum of distances between  $u$ and the pendent vertices of $T$ where $u \in V(T)$.
We start with a lemma given below.
\begin{lemma}\label{lemma1}
Let  $T_{a}$ and $T_{b}$ be two trees of orders   $n_{1}$ and $n_{2}$   respectively.  Let $u \in V(T_{a})$, $v \in V(T_{b})$ and  $T_{a}.T_{b} (u,v)$ denotes the new tree obtained from $T_{a}$ and $T_{b}$ by identifying $u$ and $v$. Let the number of pendent vertices  in $T_{a}$ and $T_{b}$ be $l_{1}$ and $l_{2}$ respectively.\\
If both $u$ and $v$ are nonpendent vertices, then 
\begin{equation}\label{eq:6}
TW(T_{a}.T_{b} (u,v))=TW(T_{a})+TW(T_{b})+l_{2}d^{+}(u)+ l_{1}d^{+}(v). 
\end{equation}
\end{lemma}
The above lemma can be extended to three trees as follows.
\begin{lemma}\label{lemma2}
Let  $T_{a}$,$T_{b}$ and $T_{c}$ be three trees of orders   $n_{1}$,$n_{2}$ and $n_{3}$   respectively.  Let $u \in V(T_{a})$, $v \in V(T_{b})$,$w \in V(T_{c})$ and  $T_{a}.T_{b}.T_{c} (u,v,w)$ denotes the new tree obtained from $T_{a}$, $T_{b}$ and $T_{c}$ by identifying $u$, $v$ and $w$. Let the number of pendent vertices  in $T_{a}$, $T_{b}$ and $T_{c}$ be $l_{1}$,$l_{2}$ and $l_{3}$ respectively.\\
If $u$,$v$ and $w$ are nonpendent vertices, then 
\begin{equation}\label{eq:7}
TW(T_{a}.T_{b}.T_{c} (u,v,w))=TW(T_{a})+TW(T_{b})+TW(T_{c})+(l_{2}+l_{3})d^{+}(u)+ (l_{1}+l_{3})d^{+}(v) +(l_{1}+l_{2})d^{+}(w). 
\end{equation}
\end{lemma}
\begin{lemma}\label{lemma3}
Let $T$ be a tree  composed of two disjoint trees $T_{a}$ and $T_{b}$ of orders   $n_{1}$ and $n_{2}$   respectively.  Let $u \in V(T_{a})$, $v \in V(T_{b})$ and  $uv$ be   a cut-edge in $T$. Let the number of leaves in $T_{a}$ and $T_{b}$ be $l_{1}$ and $l_{2}$ respectively.
In $T$, for a vertex $x\in V(T)$ let $d^{+}(x)$ denote the  sum of  the distances of the form $d(x,l)$ where $l$ is a leaf in $T$.
Then
\begin{equation}\label{eq:3}
TW(T)=TW(T_{a})+TW(T_{b})+l_{2}d^{+}(u)+ l_{1}d^{+}(v)+l_{1}l_{2}. \\
\end{equation}
\end{lemma} 
Equiseparability w.r.t the Wiener index  and equiseparability  w.r.t the terminal Wiener index are entirely different. So, it is useful to find  some  rules for  constructing equiseparable trees w.r.t the terminal Wiener index only. The following theorem can be  used to construct equiseparable  trees with respect to terminal Wiener index.
\begin{thm}\label{thm1}(\cite{xia})
Let $T$, $X$ and $Y$ be arbitary trees, each with more than two vertices. Let  $T_{1}$ be obtained from  $T$ by identifying the vertices  $u$ and $s$ and by identifying the vertices $v$ and $r$. Let  $T_{2}$ be obtained from  $T$ by identifying the vertices  $u$ and $r$, and by identifying the vertices $v$ and $s$. If $p_{x}-p_{s}=p_{y}-p_{r}$, then $T_{1}$ and $T_{2}$ are equiseparable  w.r.t terminal  Wiener index. $p_{x}$ and $p_{y}$ denote the number of pendent vertices of  fragments $X$ and $Y$, respectively. $p_{s}=1$ if $s$ is a pendent vertex of $X$;otherwise it is equal to 0. $p_{r}$ is defined similar to $p_{s}$.  
\end{thm}
\begin{figure}[H]
\begin{center}
\begin{pspicture}[showgrid=false](0,-3)(10,4)


\psdots[linewidth=8pt,dotsize=3pt 0](1,4)
\psdots[linewidth=8pt,dotsize=3pt 0](3,4)
\psdots[linewidth=8pt,dotsize=3pt 0](6,2)
\psdots[linewidth=8pt,dotsize=3pt 0](9,2)
\psdots[linewidth=8pt,dotsize=3pt 0](1,0)
\psdots[linewidth=8pt,dotsize=3pt 0](3,0)
\psdots[linewidth=8pt,dotsize=3pt 0](7,0)
\psdots[linewidth=8pt,dotsize=4pt 0](9,0)

\fontsize{8}{10}\rput[bl](2,3){T}
\fontsize{8}{10}\rput[bl](5.8,2.5){X}
\fontsize{8}{10}\rput[bl](8.8,2.5){Y}
\fontsize{8}{10}\rput[bl](2,-2.5){$T_{1}$}
\fontsize{8}{10}\rput[bl](8,-2.5){$T_{2}$}
\fontsize{8}{10}\rput[bl](2,-1){T}
\fontsize{8}{10}\rput[bl](8,-1){T}
\fontsize{8}{10}\rput[bl](.85,.5){X}
\fontsize{8}{10}\rput[bl](2.85,.5){Y}
\fontsize{8}{10}\rput[bl](6.85,.5){Y}
\fontsize{8}{10}\rput[bl](8.85,.5){X}

\psline (0,4)(4,4)
\psline (0,4)(2,2)
\psline (2,2)(4,4)
\psline (6,2)(5.5,3)
\psline (5.5,3)(6.5,3)
\psline (6,2)(6.5,3)
\psline (9,2)(8.5,3)
\psline (8.5,3)(9.5,3)
\psline (9,2)(9.5,3)

\psline (0,0)(4,0)
\psline (0,0)(2,-2)
\psline (2,-2)(4,0)

\psline (6,0)(10,0)
\psline (6,0)(8,-2)
\psline (8,-2)(10,0)

\psline (1,0)(.5,1)
\psline (0.5,1)(1.5,1)
\psline (1.5,1)(1,0)

\psline (3,0)(2.5,1)
\psline (2.5,1)(3.5,1)
\psline (3.5,1)(3,0)

\psline (7,0)(6.5,1)
\psline (6.5,1)(7.5,1)
\psline (7.5,1)(7,0)

\psline (9,0)(8.5,1)
\psline (8.5,1)(9.5,1)
\psline (9.5,1)(9,0)

\rput[bl](1,4.1){$u$}
\rput[bl](3,4.1){$v$}
\rput[bl](6,1.8){$s$}
\rput[bl](9,1.8){$r$}
\rput[bl](1,-0.2){$u$}
\rput[bl](3,-0.2){$v$}
\rput[bl](7,-0.2){$u$}
\rput[bl](9,-0.2){$v$}
\end{pspicture}
\caption{Two  trees $T_{1}$ and $T_{2}$ constructed by method in theorem \ref{thm1} }\label{Fig.1.}
\end{center}
\end{figure}
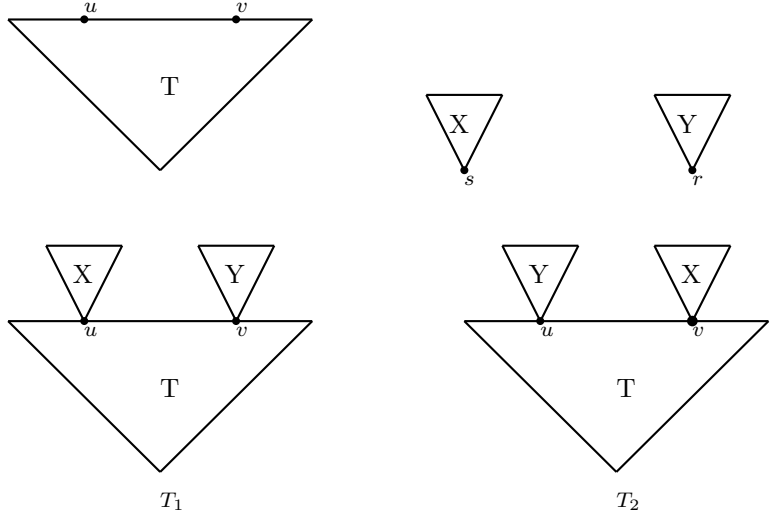
Fig.\ref{Fig.1.} shows the construction of equiseparable trees w.r.t terminal Wiener idex.
We will prove that the condition in theorem \ref{thm1} is sufficient but not necessary  to generate  equiseparable trees w.r.t terminal  Wiener index.
 We extend  theorem \ref{thm1} as
\begin{thm}\label{thm2}
Let $T$, $X$ and $Y$ be arbitary trees, each with at least three vertices. Let  $T_{3}$ be obtained from  $T$ by identifying the vertices  $u$ and $s$ and by identifying the vertices $v$ and $r$. Let  $T_{4}$ be obtained from  $T$ by identifying the vertices  $u$ and $r$, and by identifying the vertices $v$ and $s$. If \textbf{either $p_{x}-p_{s}=p_{y}-p_{r}$ or $d^{+}(u)=d^{+}(v)$}, then $T_{3}$ and $T_{4}$ are equiseparable  w.r.t terminal  Wiener index. $p_{x}$,$p_{y}$,$p_{t}$,$p_{s}$ and $p_{r}$are same as in theorem \ref{thm1}.
\end{thm}
\begin{proof}
By lemma \ref{lemma2}, we get
\begin{equation}\label{eq:7}
\begin{split}
TW(T_{1})=TW(T)+TW(X)+TW(Y)+d^{+}(s)(p_{t}+p_{y}-p_{r})+ d^{+}(r)(p_{t}+p_{x}-p_{s})\\+(p_{x}-p_{s})(p_{y}-p_{r})d(u,v)+d^{+}(u)(p_{x}-p_{s})+d^{+}(v)(p_{y}-p_{r})-p_{s}d^{+}(s)-p_{r}d^{+}(r)
\end{split}
\end{equation}
 and 
\begin{equation}\label{eq:8}
\begin{split}
TW(T_{2})=TW(T)+TW(X)+TW(Y)+d^{+}(s)(p_{t}+p_{y}-p_{r})+ d^{+}(r)(p_{t}+p_{x}-p_{s})+\\(p_{x}-p_{s})(p_{y}-p_{r})d(u,v)+d^{+}(u)(p_{y}-p_{r})+d^{+}(v)(p_{x}-p_{s})-p_{s}d^{+}(s)-p_{r}d^{+}(r).
\end{split} 
\end{equation}
where the last two terms in (\ref{eq:7}) and (\ref{eq:8}) respectively are correction factors for $TW(X)$ and $TW(Y)$ in case $s$ and $r$ are pendent. From (\ref{eq:7}) and (\ref{eq:8}), for $TW(T_{1})=TW(T_{2})$\\
$d^{+}(u)(p_{x}-p_{s})+d^{+}(v)(p_{y}-p_{r})=d^{+}(u)(p_{y}-p_{r})+d^{+}(v)(p_{x}-p_{s}).$
This will be true if either (a) $p_{x}-p_{s}=p_{y}-p_{r}$ or (b) $d^{+}(u)=d^{+}(v)$ holds.  
\end{proof}
The above theorem  shows that  it is possible to   construct  equiseparable trees  w.r.t  terminal Wiener index based on distance parameters also.
Consider the following trees $T_{3}$ and $T_{4}$.\\
\begin{figure}[H]
\begin{center}
\begin{pspicture}[showgrid=false](0,0)(11,5)
\cnode*(0,5){0.10}{A}
\psline(0.1,5)(.9,5)
\cnode*(1,5){.1}{A}
\psline(1.1,5)(1.9,5)
\cnode*(2,5){.1}{A}
\psline(2.1,5)(2.9,5)
\cnode*(3,5){.1}{A}
\psline(3.1,5)(3.9,5)
\cnode*(4,5){.1}{A}

\rput[bl](2,5.2){u}
\rput[bl](3,5.2){$v$}
\rput[bl](6,3.7){$s$}
\rput[bl](9,3.65){$r$}
\rput[bl](1.9,3.3){$T$}
\rput[bl](5.9,3.3){$X$}
\rput[bl](8.9,3.3){$Y$}

\cnode*(2,4){.1}{A}
\cnode*(3,4){.1}{A}
\psline(2,4)(2,5)
\psline(3,4)(3,5)
\cnode*(6,4){.1}{A}
\cnode*(5,4){.1}{A}
\cnode*(7,4){.1}{A}
\cnode*(5.5,4.75){.1}{A}
\cnode*(6.5,4.75){.1}{A}
\psline(5,4)(6,4)
\psline(6,4)(7,4)
\psline(6,4)(5.5,4.75)
\psline(6,4)(6.5,4.75)
\cnode*(9,4){.1}{A}
\cnode*(9.75,4.75){.1}{A}
\psline(9,4)(9.75,4.75)
\cnode*(9,5){.1}{A}
\psline(9,4)(9,5)

\cnode*(8.25,5.75){.1}{A}
\cnode*(9.75,5.75){.1}{A}
\psline(8.25,5.75)(9,5)
\psline(9.75,5.75)(9,5)

\cnode*(0,1){0.10}{A}
\psline(0.1,1)(.9,1)
\cnode*(1,1){.1}{A}
\psline(1.1,1)(1.9,1)
\cnode*(2,1){.1}{A}
\psline(2.1,1)(2.9,1)
\cnode*(3,1){.1}{A}
\psline(3.1,1)(3.9,1)
\cnode*(4,1){.1}{A}
\cnode*(2,0){.1}{A}
\cnode*(3,0){.1}{A}
\psline(2,0)(2,1)
\psline(3,0)(3,1)

\cnode*(6,1){0.10}{A}
\psline(6.1,1)(6.9,1)
\cnode*(7,1){.1}{A}
\psline(7.1,1)(7.9,1)
\cnode*(8,1){.1}{A}
\psline(8.1,1)(8.9,1)
\cnode*(9,1){.1}{A}
\psline(9.1,1)(9.9,1)
\cnode*(10,1){.1}{A}
\cnode*(8,0){.1}{A}
\cnode*(9,0){.1}{A}
\psline(8,0)(8,1)
\psline(9,0)(9,1)

\cnode*(2,1){.1}{A}
\cnode*(1.25,1.25){.1}{A}
\cnode*(1.5,1.75){.1}{A}
\cnode*(2.5,1.75){.1}{A}
\cnode*(2.75,1.25){.1}{A}
\psline(2,1)(1.25,1.25)
\psline(2,1)(1.5,1.75)
\psline(2,1)(2.5,1.75)
\psline(2,1)(2.75,1.25)

\cnode*(3,1){.1}{A}
\cnode*(3.75,1.75){.1}{A}
\psline(3,1)(3.75,1.75)
\cnode*(3,2){.1}{A}
\psline(3,1)(3,2)
\cnode*(3.75,2.75){.1}{A}
\psline(3.75,2.75)(3,2)
\cnode*(2.25,2.75){.1}{A}
\psline(2.25,2.75)(3,2)

\cnode*(8,2){.1}{A}
\cnode*(7.25,2.75){.1}{A}
\psline(8,1)(8,2)
\cnode*(8.75,2.75){.1}{A}
\psline(7.25,2.75)(8,2)
\psline(8.75,2.75)(8,2)
\cnode*(7.25,1.75){.1}{A}
\psline(8,1)(7.25,1.75)

\cnode*(9,1){.1}{A}
\cnode*(8.25,1.25){.1}{A}
\cnode*(8.5,1.75){.1}{A}
\cnode*(9.5,1.75){.1}{A}
\cnode*(9.75,1.25){.1}{A}
\psline(9,1)(8.25,1.25)
\psline(9,1)(8.5,1.75)
\psline(9,1)(9.5,1.75)
\psline(9,1)(9.75,1.25)
\fontsize{8}{10}\rput[bl](2.5,-.5){$T_{3}$}
\fontsize{8}{10}\rput[bl](8.5,-.5){$T_{4}$}
\end{pspicture}
\end{center}
\caption{Equiseparable trees w.r.t terminal    Wiener index constructed by method in theorem \ref{thm2}. }\label{fig.2.}
\end{figure}
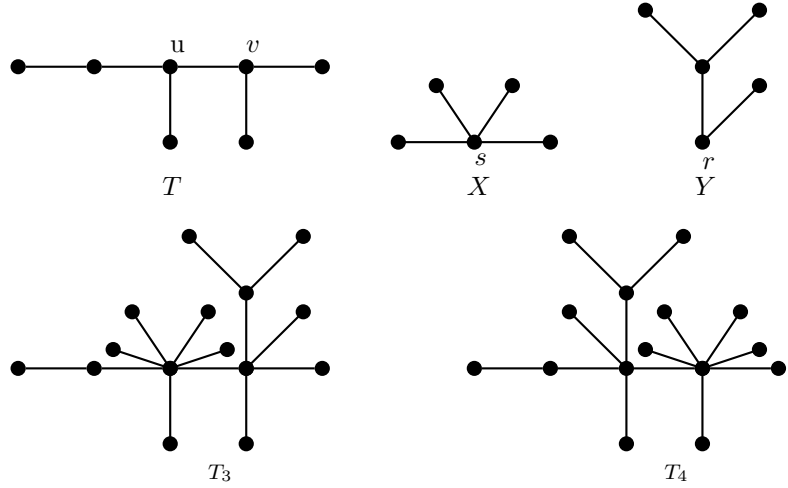
In Fig.\ref{fig.2.} we can see that $p_{x}-p_{s} \neq p_{y}-p_{r}$ but  $TW(T_{3})=TW(T_{4})$. For vertices $u$ and $v$, $d^{+}(u)=d^{+}(v)$.
The following  theorem is based on partitioning of integers.
\begin{thm}\label{thm5}
Let $T$ be a tree  of order $n$. Let $n_{1},n_{2}$ and $n_{3},n_{4}$ be  two different partitions of $n$. Let $T'$ be tree obtained by connecting some nonpendent vertex of $P_{n_{1}}$ and $P_{n_{2}}$ by an edge. Let $T''$ be tree obtained by connecting some nonpendent vertex of $P_{n_{3}}$ and $P_{n_{4}}$ by an edge. Then $T'$ and $T''$ are equiseparable w.r.t terminal Wiener index.
\end{thm}
\begin{proof}
Since $TW(P_{n_{1}})+TW(P_{n_{2}})=TW(P_{n_{3}})+TW(P_{n_{4}})=n-2$ and $l_{1}=l_{2}=2$ and $d^{+}(u)+d^{+}(v)=n-2$, by lemma \ref{lemma3} both $T'$ and $T''$ are equiseparable w.r.t terminal Wiener index.
\end{proof}
Consider two trees $T'$ and $T''$ constructed using theorem $\ref{thm5}$.
\begin{figure}[H]
\begin{center}
\begin{pspicture}[showgrid=false](-1,-1)(11,3)
\cnode*(0,2){0.10}{A}
\psline(0.1,2)(.9,2)
\cnode*(1,2){.1}{A}
\psline(1.1,2)(1.9,2)
\cnode*(2,2){.1}{A}
\psline(2.1,2)(2.9,2)
\cnode*(3,2){.1}{A}
\psline(3.1,2)(3.9,2)
\cnode*(4,2){.1}{A}
\psline(4.1,2)(4.9,2)
\cnode*(5,2){.1}{A}
\cnode*(2,1){.1}{A}
\cnode*(3,1){.1}{A}
\psline(2,1.9)(2,1.1)
\psline(3,1.9)(3,1.1)
\psset{linestyle=solid}
\cnode*(6,2){0.10}{A}
\psline(6.1,2)(6.9,2)
\cnode*(7,2){.1}{A}
\psline(7.1,2)(7.9,2)
\cnode*(8,2){.1}{A}
\psline(8.1,2)(8.9,2)
\cnode*(9,2){.1}{A}
\psline(9.1,2)(9.9,2)
\cnode*(7,0){.1}{A}
\cnode*(7,2){.1}{A}
\cnode*(9,0){.1}{A}
\cnode*(10,2){.1}{A}
\cnode*(8,1){.1}{A}
\psline(8,1.1)(8,1.9)
\psline(7.05,0)(7.95,0.9)
\rput[bl](2,-1){$T'$}
\psline(8.05,.95)(8.95,0.05)
\rput[bl](8,-1){$T''$}
\end{pspicture}
\end{center}
\caption{Two trees equiseparable w.r.t  terminal Wiener index. } \label{fig:3}
\end{figure}
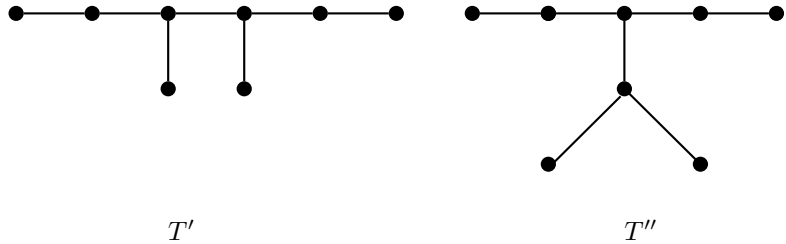
In fig.\ref{fig:3} $n_{1}=4$,$n_{2}=4$, $n_{3}=5$ and $n_{4}=3$. It is easy to see that $TW(T')=TW(T'')=22$.\\
In the rest of the paper, we use $p_{t}$ to denote the number of pendent vertices in a tree $T$.
\begin{thm}\label{thm3}
Let $Z$ be an  arbitary tree; $u\in V(Z)$, tree   $T'$ is  obtained from $T$ by identifying the vertices  $u$ and $i$ and  $T''$ be obtained from $T$ by identifying the vertices  $u$ and $j$. If $d^{+}(i)=d^{+}(j)$, then  $T'$ and $T''$ are equiseparable w.r.t terminal Wiener index. See Fig. \ref{fig.4.}
\end{thm}
\begin{figure}[H]
\begin{center}
\begin{pspicture}[showgrid=false](0,0)(11,5)
\psline(2,4)(1,5)
\psline(2,4)(1,3)
\psline(1,3)(1,5)
\psline(2,4)(9,4)
\psline(10,5)(9,4)
\psline(10,3)(9,4)
\psline(10,3)(10,5)

\rput[bl](1.25,4){X}
\rput[bl](9.4,4){Y}
\rput[bl](1.25,1){X}
\rput[bl](9.4,1){Y}
\rput[bl](3.75,4.5){Z}
\rput[bl](6.75,1.5){Z}

\rput[bl](2,3.6){$1$}
\rput[bl](2,0.6){$1$}
\rput[bl](4,0.6){$i$}
\rput[bl](7,.6){$j$}
\rput[bl](4,3.6){$i$}
\rput[bl](7,3.6){$j$}
\rput[bl](3.9,4.2){$u$}
\rput[bl](6.9,1.2){$u$}
\rput[bl](2.75,3.7){$.$}
\rput[bl](3,3.7){$.$}
\rput[bl](3.25,3.7){$.$}

\rput[bl](5.25,3.7){$.$}
\rput[bl](5.5,3.7){$.$}
\rput[bl](5.75,3.7){$.$}

\rput[bl](7.75,3.7){$.$}
\rput[bl](8,3.7){$.$}
\rput[bl](8.25,3.7){$.$}

\rput[bl](2.75,.7){$.$}
\rput[bl](3,.7){$.$}
\rput[bl](3.25,.7){$.$}

\rput[bl](5.25,.7){$.$}
\rput[bl](5.5,.7){$.$}
\rput[bl](5.75,.7){$.$}

\rput[bl](7.75,.7){$.$}
\rput[bl](8,.7){$.$}
\rput[bl](8.25,.7){$.$}

\rput[bl](8.8,0.6){$m$}
\rput[bl](8.8,3.6){$m$}

\cnode*(2,4){.1}{A}
\cnode*(9,4){.1}{A}

\cnode*(2,1){.1}{A}
\cnode*(9,1){.1}{A}
\psline(2,1)(1,2)
\psline(2,1)(1,0)
\psline(1,0)(1,2)
\psline(2,1)(9,1)
\psline(10,0)(9,1)
\psline(10,2)(9,1)
\psline(10,0)(10,2)

\psline(4,4)(5,5)
\psline(4,4)(3,5)
\psline(3,5)(5,5)
\psline(7,1)(6,2)
\psline(7,1)(8,2)
\psline(6,2)(8,2)

\cnode*(4,1){.1}{A}
\cnode*(4,4){.1}{A}
\cnode*(7,1){.1}{A}
\cnode*(7,4){.1}{A}

\end{pspicture}
\end{center}
\caption{Two  trees  $T'$ and $T''$ equiseparable w.r.t  terminal   Wiener index. }\label{fig.4.}
\end{figure}
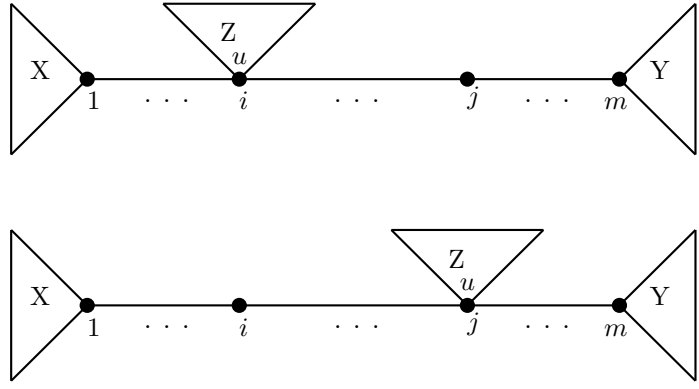
\begin{proof}
By lemma \ref{lemma1}, we get\\
\begin{equation}\label{eq:9}
TW(T')=TW(T)+TW(Z)+d^{+}(i)p_{z}+d^{+}(u)p_{t}
\end{equation}
 and 
\begin{equation}\label{eq:10}
TW(T'')=TW(T)+TW(Z)+d^{+}(j)p_{z}+d^{+}(u)p_{t}.
\end{equation} where $p_{t}$ and $p_{z}$ denote the number of pendent vertices in $T$ and $Z$ respectively.\\
From (\ref{eq:9}) and (\ref{eq:10}),
for $TW(T')=TW(T'')$  the equality
$d^{+}(i)=d^{+}(j)$ must be true.
\end{proof}
\begin{thm}\label{thm4}
Let $T$ be a tree with vertices $i_{1},i_{2},....i_{n}$. Let $Z_{a}$'s be arbitary trees for ${1 \leq a \leq  t}$. Let $T_{1}$ be the tree obtained by identifying $i_{a}$ and $u_{a}$ in $Z_{a}$. $T_{2}$ is obtained by identifying $i_{a}+j$ and $u_{a}$ in $Z_{a}$.  If $d^{+}(i_{a})=d^{+}(i_{a}+j)$ for every $1 \leq a \leq  t$, then the resulting two trees are equiseparable  w.r.t terminal  Wiener index. See Fig. \ref{fig.5.} 
\end{thm}
\begin{figure}[H]
\begin{center}
\begin{pspicture}[showgrid=false](0,0)(14,5)
\psline(2,4)(1,5)
\psline(2,4)(1,3)
\psline(1,3)(1,5)
\psline(2,4)(12,4)
\psline(13,5)(12,4)
\psline(13,3)(12,4)
\psline(13,3)(13,5)

\rput[bl](1.25,4){X}
\rput[bl](12.4,4){Y}
\rput[bl](1.25,1){X}
\rput[bl](12.4,1){Y}
\rput[bl](3.75,4.5){$Z_{1}$}
\rput[bl](6.75,4.5){$Z_{2}$}
\rput[bl](3.75,1.5){$Z_{1}$}
\rput[bl](6.75,1.5){$Z_{2}$}
\rput[bl](9.75,4.5){$Z_{t}$}
\rput[bl](9.75,1.5){$Z_{t}$}

\rput[bl](2,3.6){$1$}
\rput[bl](2,0.6){$1$}
\rput[bl](3.7,0.5){$i_{1}+j$}
\rput[bl](6.7,.5){$i_{2}+j$}
\rput[bl](4,3.5){$i_{1}$}
\rput[bl](7,3.5){$i_{2}$}
\rput[bl](3.8,4.2){$u_{1}$}
\rput[bl](3.8,1.2){$u_{1}$}
\rput[bl](6.8,1.2){$u_{2}$}
\rput[bl](6.8,4.2){$u_{2}$}
\rput[bl](9.8,4.2){$u_{t}$}
\rput[bl](9.8,1.2){$u_{t}$}
\rput[bl](2.75,3.7){$.$}
\rput[bl](3,3.7){$.$}
\rput[bl](3.25,3.7){$.$}

\rput[bl](5.25,3.7){$.$}
\rput[bl](5.5,3.7){$.$}
\rput[bl](5.75,3.7){$.$}

\rput[bl](8.25,3.7){$.$}
\rput[bl](8.5,3.7){$.$}
\rput[bl](8.75,3.7){$.$}

\rput[bl](11.8,.6){$k$}
\rput[bl](11.8,3.6){$k$}

\rput[bl](11,3.7){$.$}
\rput[bl](11.25,3.7){$.$}
\rput[bl](11.5,3.7){$.$}

\rput[bl](11,.7){$.$}
\rput[bl](11.25,.7){$.$}
\rput[bl](11.5,.7){$.$}

\rput[bl](2.75,.7){$.$}
\rput[bl](3,.7){$.$}
\rput[bl](3.25,.7){$.$}

\rput[bl](5.25,.7){$.$}
\rput[bl](5.5,.7){$.$}
\rput[bl](5.75,.7){$.$}

\rput[bl](8.25,.7){$.$}
\rput[bl](8.5,.7){$.$}
\rput[bl](8.75,.7){$.$}

\rput[bl](9.8,0.5){$i_{t}+j$}
\rput[bl](9.8,3.5){$i_{t}$}

\cnode*(2,4){.1}{A}
\cnode*(12,4){.1}{A}

\cnode*(2,1){.1}{A}
\cnode*(12,1){.1}{A}
\psline(2,1)(1,2)
\psline(2,1)(1,0)
\psline(1,0)(1,2)
\psline(2,1)(12,1)
\psline(13,0)(12,1)
\psline(13,2)(12,1)
\psline(13,0)(13,2)

\psline(4,4)(5,5)
\psline(4,4)(3,5)
\psline(3,5)(5,5)
\psline(7,1)(6,2)
\psline(7,1)(8,2)
\psline(6,2)(8,2)

\cnode*(4,1){.1}{A}
\cnode*(4,4){.1}{A}
\cnode*(7,1){.1}{A}
\cnode*(7,4){.1}{A}
\psline(7,4)(6,5)
\psline(8,5)(6,5)
\psline(8,5)(7,4)
\cnode*(10,4){.1}{A}
\cnode*(10,1){.1}{A}
\psline(9,5)(10,4)
\psline(11,5)(10,4)
\psline(9,5)(11,5)
\psline(4,1)(3,2)
\psline(4,1)(5,2)
\psline(5,2)(3,2)
\psline(10,1)(9,2)
\psline(10,1)(11,2)
\psline(11,2)(9,2)
\end{pspicture}
\end{center}
\caption{Two  trees equiseparable w.r.t  terminal   Wiener index. }\label{fig.5.}
\end{figure}
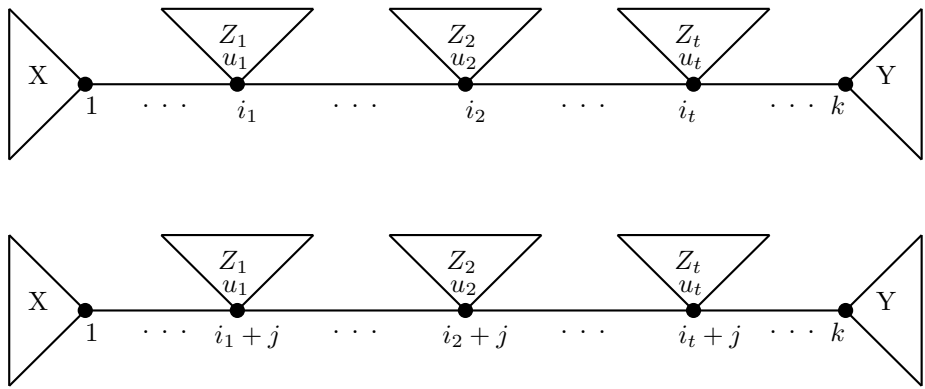
\begin{proof}
By lemma \ref{lemma1}, we get
\begin{equation}\label{eq:11}
TW(T_{1})=TW(T)+\sum\limits_{i =1}^{t}TW(Z_{i})+\sum\limits_{1 \leq a \leq  t} (p_{z_{a}}d^{+}(i_{a})+p_{t}d^{+}(u_{a}))
\end{equation}
 and 
\begin{equation}\label{eq:12}
TW(T_{2})=TW(T)+\sum\limits_{i =1}^{t}TW(Z_{i})+\sum\limits_{1 \leq a \leq  t} (p_{z_{a}}d^{+}(i_{a}+j)+p_{t}d^{+}(u_{a})).
\end{equation}
for some $j$.
For $TW(T_{1})=TW(T_{2})$ it must be true that 
$d^{+}(i_{a})=d^{+}(i_{a}+j)$ for every $1 \leq a \leq  t $ and for some $j$.
\end{proof}

\section {Large families of   trees with same  terminal  Wiener index}
Let $P_{2k}$ be a path on $2k$ vertices labelled as $v_{1},v_{2}...v_{2k}$. Let $X$ and $Y$ be two different trees with same number of vertices. Let $s \in V(X)$ and $t \in V(Y)$ such that $p_{x}-p_{s}=p_{y}-p_{t}$ where $p_{x}$   and $p_{y}$ denote the number of pendent vertices in $X$ and $Y$ respectively. $p_{s}=1$ if $s$ is a pendent vertex of $X$; otherwise it is equal to $0$. $p_{t}$ is defined similar to $p_{s}$. We can construct  a set STW(T,X,Y) of trees equiseparable with respect to  terminal  Wiener index as follows.\\
Each element of the set is obtained from  a copy of $T$, $k$ copies of $X$ and $k$ copies of $Y$. Fragments  of $X$ are attached(via their vertices $s$) to $k$ among the vertices $v_{1},v_{2}...v_{2k}$ of $T$. Fragments  of $Y$ are attached(via their vertices $t$)  to the remaining $k$   vertices among $v_{1},v_{2}...v_{2k}$ of $T$.\\

\begin{thm}\label{thm6}
The trees in the set $STW(T,X,Y)$ have the same terminal Wiener index.
\end{thm}
\begin{proof}
Let $T'$  and $T''$ be two trees belonging to the set $STW(T,X,Y)$. It is easy to see that $p_{u}(e_{i}|T')=p_{u}(e_{i}|T'')$  for corresponding edges  of $X$ and $Y$ in $T'$ and $T''$ respectively. Considering the corresponding edges of $P_{2k}$ in $T'$ and $T''$, we can see that
\begin{flushleft}  
$p_{u}(v_{i}v_{i+1}|T')=p_{u}(v_{i}v_{i+1}|T'')=ik^{2}(2k-1).$
\end{flushleft}  
Thus, for every pair of corresponding edges  $p_{u}(e_{i}|T')=p_{u}(e_{i}|T'')$ is true.
\end{proof}
Fig. \ref{fig.6.}  shows two trees $T_{5}$ and $T_{6}$ constructed from trees $T$,$X$ and $Y$ using the above method. Here $k=2,p_{s}=1$ and $p_{t}=0$.
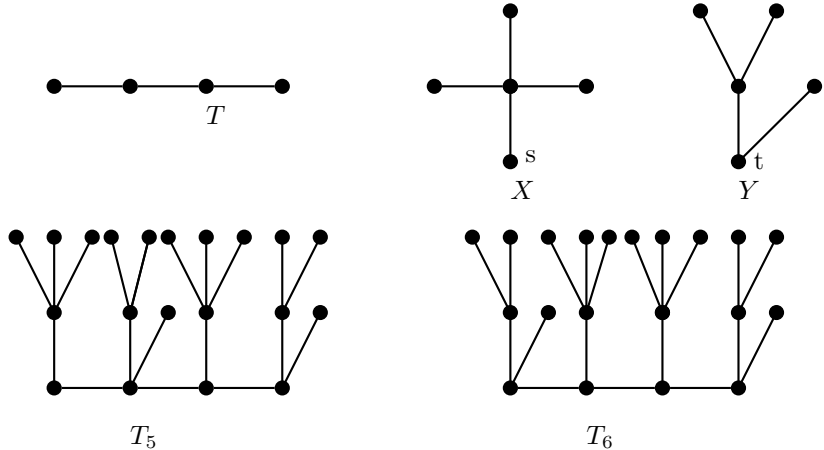
\begin{figure}[H]
\begin{center}
\begin{pspicture}[showgrid=false](0,-3)(11,3)
\cnode*(1,2){.1}{A}
\psline(1.1,2)(1.9,2)
\cnode*(2,2){.1}{A}
\psline(2.1,2)(2.9,2)
\cnode*(3,2){.1}{A}
\psline(3.1,2)(3.9,2)
\cnode*(4,2){.1}{A}

\cnode*(1,-2){.1}{A}
\psline(1.1,-2)(1.9,-2)
\cnode*(2,-2){.1}{A}
\psline(2.1,-2)(2.9,-2)
\cnode*(3,-2){.1}{A}
\psline(3.1,-2)(3.9,-2)
\cnode*(4,-2){.1}{A}
\cnode*(3,-1){.1}{A}
\psline(3,-2)(3,-1)
\cnode*(1,-1){.1}{A}
\psline(1,-2)(1,-1)
\cnode*(1,0){.1}{A}
\psline(1,-1)(1,0)
\cnode*(1.5,0){.1}{A}
\psline(1,-1)(1.5,0)
\cnode*(.5,0){.1}{A}
\psline(1,-1)(.5,0)
\cnode*(2,-1){.1}{A}
\psline(2,-1)(2,-2)
\cnode*(1.75,0){.1}{A}
\cnode*(2.25,0){.1}{A}
\psline(2,-1)(1.75,0)
\psline(2,-1)(2.25,0)
\psline(2,-1)(2.25,0)
\cnode*(2.5,-1){.1}{A}
\psline(2,-2)(2.5,-1)
\cnode*(3,0){.1}{A}
\psline(3,-1)(3,0)
\cnode*(2.5,0){.1}{A}
\cnode*(4,0){.1}{A}
\psline(2.5,0)(3,-1)
\psline(4,-1)(4,0)
\cnode*(4,-1){.1}{A}
\cnode*(4.5,-1){.1}{A}
\psline(4,-2)(4,-1)
\psline(4,-2)(4.5,-1)
\cnode*(3.5,0){.1}{A}
\cnode*(4.5,0){.1}{A}
\psline(3.5,0)(3,-1)
\psline(4.5,0)(4,-1)

\cnode*(7,-2){.1}{A}
\psline(7.1,-2)(7.9,-2)
\cnode*(8,-2){.1}{A}
\psline(8.1,-2)(8.9,-2)
\cnode*(9,-2){.1}{A}
\psline(9.1,-2)(9.9,-2)
\cnode*(10,-2){.1}{A}
\cnode*(9,-1){.1}{A}
\psline(9,-2)(9,-1)

\psset{linestyle=solid}
\cnode*(6,2){0.10}{A}
\psline(6.1,2)(6.9,2)
\cnode*(7,2){.1}{A}
\psline(7.1,2)(7.9,2)
\cnode*(8,2){.1}{A}

\cnode*(9.5,3){.1}{A}
\cnode*(10.5,3){.1}{A}
\psline(10,2)(9.5,3)
\psline(10,2)(10.5,3)

\cnode*(7,3){.1}{A}
\cnode*(7,2){.1}{A}
\cnode*(10,2){.1}{A}
\cnode*(7,1){.1}{A}
\psline(7,1)(7,2)
\psline(7,2)(7,3)
\rput[bl](7.2,1){s}
\rput[bl](8,-2.8){$T_{6}$}
\cnode*(10,1){.1}{A}
\psline(10,1)(11,2)
\psline(10,1)(10,2)
\cnode*(11,2){.1}{A}
\cnode*(7,-1){.1}{A}
\psline(7,-1)(7,-2)
\cnode*(7.5,-1){.1}{A}
\psline(7.5,-1)(7,-2)
\cnode*(6.5,0){.1}{A}
\cnode*(7,0){.1}{A}
\psline(7,-1)(6.5,0)
\psline(7,-1)(7,0)
\cnode*(8,-1){.1}{A}
\psline(8,-1)(8,-2)
\cnode*(7.5,0){.1}{A}
\cnode*(8,0){.1}{A}
\cnode*(8.3,0){.1}{A}
\psline(8,-1)(7.5,0)
\psline(8,-1)(8,0)
\psline(8,-1)(8.3,0)
\cnode*(9,-1){.1}{A}
\cnode*(9,0){.1}{A}
\psline(9,-1)(9,0)
\cnode*(8.6,0){.1}{A}
\psline(9,-1)(9.5,0)
\psline(9,-1)(8.6,0)
\cnode*(10,-1){.1}{A}
\cnode*(10.5,-1){.1}{A}
\psline(10,-2)(10,-1)
\psline(10,-2)(10.5,-1)
\cnode*(10,0){.1}{A}
\cnode*(9.5,0){.1}{A}
\cnode*(10.5,0){.1}{A}
\psline(10,-1)(10,0)
\psline(10,-1)(10.5,0)
\rput[bl](3,1.5){$T$}
\rput[bl](7,.5){$X$}
\rput[bl](10,.5){$Y$}

\rput[bl](10.2,.9){t}
\rput[bl](2,-2.8){$T_{5}$}

\end{pspicture}
\end{center}
\caption{Two  trees equiseparable with respect to   terminal Wiener index. } \label{fig.6.}
\end{figure}

We use the notation  $T.T'(u,v)$ to denote the tree obtained  from  $T$ and $T'$ by identifying the vertices  $u$  and $v$.
\begin{thm}\label{thm7}
Let $T$ be a tree with two  vertices $u$ and $v$ such that $d^{+}(u)=d^{+}(v)$. Let $T'$ be another tree with a   vertex $x$. Then $T.T'(u,x)$ and $T.T'(v,x)$ are equiseparable w.r.t terminal Wiener index.
\end{thm}
\begin{proof}
Substituting   $T_{a}=T$, $T_{b}=T'$ and $d^{+}(u)=d^{+}(v)$ in lemma \ref{lemma1} we get  the result.
\end{proof}
\begin{figure}[H]
\begin{center}
\begin{pspicture}[showgrid=false](-1,-1)(11,3)
\cnode*(0,2){0.10}{A}
\cnode*(-1,2){.1}{A}
\psline(-1,2)(0,2)
\psline(0.1,2)(.9,2)
\cnode*(1,2){.1}{A}
\psline(1.1,2)(1.9,2)
\cnode*(2,2){.1}{A}
\psline(2.1,2)(2.9,2)
\cnode*(3,2){.1}{A}
\psline(3.1,2)(3.9,2)
\cnode*(4,2){.1}{A}
\psline(4.1,2)(4.9,2)
\cnode*(5,2){.1}{A}
\cnode*(8,0){.1}{A}
\psline(8,1)(8,0)
\cnode*(3,1){.1}{A}
\psline(2,2)(2,1)

\cnode*(.5,1){.1}{A}
\psline(1,2)(.5,1)

\cnode*(1,1){.1}{A}
\cnode*(2,1){.1}{A}
\psline(1,1.9)(1,1.1)
\psline(3,1.9)(3,1.1)
\psset{linestyle=solid}
\cnode*(6,2){0.10}{A}
\psline(6.1,2)(6.9,2)
\cnode*(7,2){.1}{A}
\psline(7.1,2)(7.9,2)
\cnode*(8,2){.1}{A}
\psline(8.1,2)(8.9,2)
\cnode*(9,2){.1}{A}
\psline(9.1,2)(9.9,2)
\cnode*(7,0){.1}{A}
\cnode*(7,2){.1}{A}
\cnode*(9,0){.1}{A}
\cnode*(10,2){.1}{A}
\cnode*(8,1){.1}{A}
\psline(8,1.1)(8,1.9)
\psline(7.05,0)(7.95,1)
\rput[bl](1,1.7){u}
\rput[bl](8,2.15){v}
\rput[bl](2,-1){$T_{7}$}
\psline(8.05,1)(8.95,0)
\cnode*(9,1){.1}{A}
\psline(9,2)(9,1)
\cnode*(11,2){.1}{A}
\psline(10,2)(11,2)
\rput[bl](8.1,1){w}
\rput[bl](8,-1){$T_{8}$}
\end{pspicture}
\end{center}
\caption{Two nonisomorphic trees with the same  terminal Wiener index. } \label{fig.7.}
\end{figure}
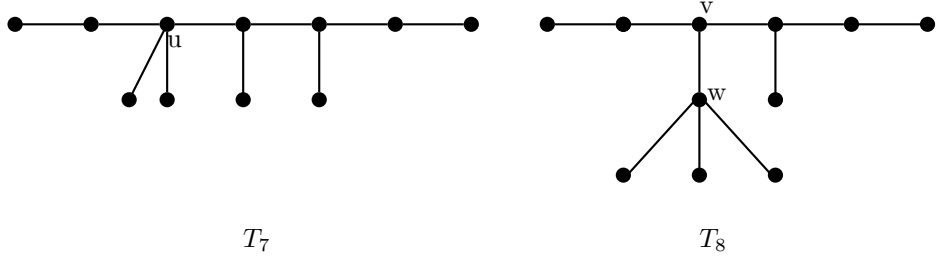
In $T_{8}$ of Fig. \ref{fig.7.}, we can note that $d^{+}(w)=d^{+}(v)=13$. Therefore by theorem \ref{thm7} we can  construct an   infinite number of non-isomorphic tree pairs  with  the same terminal Wiener index.\\

\begin{thm}\label{thm8}
Let $T_{1}$ and $T_{2}$  be two nonisomorphic trees with same terminal Wiener index. Let    $u \in V(T_{1})$ and $v\in V(T_{2})$ be two nonpendent vertices such that $d^{+}(u)=d^{+}(v)$. Let $T'$ be another tree with a   vertex $x$. Then $T_{1}.T'(u,x)$ and $T_{2}.T'(v,x)$ are nonisomorphic trees with  same terminal Wiener index.
\end{thm} 
\begin{proof}
By lemma \ref{lemma1} 
\begin{align*}
TW(T_{1}.T'(u,x))&=TW(T_{1})+TW(T')+l_{2}d^{+}(u)+ l_{1}d^{+}(x)\hspace{1cm} \text{and}\\
TW(T_{2}.T'(v,x))&=TW(T_{2})+TW(T')+l_{2}d^{+}(v)+ l_{1}d^{+}(x)
\end{align*}
Since $TW(T_{1})=TW(T_{2})$ and $d^{+}(u)=d^{+}(v)$ the result follows.
\end{proof}

For both the  non-isomorphic trees $T_{7}$ and $T_{8}$ in Fig. \ref{fig.7.}, we  have both $TW(T_{7})$ = $TW(T_{8})$ = $57$. We note that  in both these trees, $d^{+}(u)=d^{+}(w)=13$. Therefore attaching  a path at  $u$ in  $T_{7}$ and $w$ in $T_{8}$ will give rise to  another  tree pair  with  the same terminal Wiener index. By using different paths  we can  construct an   infinite number of non-isomorphic tree pairs  with  the same terminal Wiener index.
\section{Conclusion}
In this work we  proposed some new methods to generate equiseparable trees w.r.t terminal Wiener index. Equiseparable trees are mainly used in  chemistry to identify molecules with similar properties. Our future work include the study of the following two classes of trees:Trees equiseparable w.r.t Wiener index but not equiseparable w.r.t terminal Wiener index and trees equiseparable w.r.t terminal Wiener index but not equiseparable w.r.t  Wiener index.\\

\end{document}